\documentclass[letterpaper, 10pt, conference]{ieeeconf}
\IEEEoverridecommandlockouts

\usepackage{lipsum}
\usepackage{cite}
\usepackage{amsmath,amssymb,amsfonts}
\usepackage{algorithm,algorithmic}
\usepackage{hyperref}
\hypersetup{hidelinks=true}
\usepackage{textcomp}
\makeatother

\usepackage{graphicx}
\usepackage{amsthm}
\usepackage{tabu}
\usepackage{breqn}
\usepackage{verbatim}
\usepackage{esvect}
\usepackage{footnote}
\usepackage{siunitx}
\usepackage{color}
\usepackage{url}
\usepackage[makeroom]{cancel}
\usepackage{todonotes}
\usepackage{subcaption}

\newtheorem{theorem}{Theorem}
\newtheorem{corollary}[theorem]{Corollary}
\newtheorem{definition}{Definition}

\newtheorem{proposition}[theorem]{Proposition}
\newtheorem{remark}{Remark}
\newtheorem*{problem statement}{Problem Statement}

\newcommand{\real}{\mathbb{R}}

\newcommand{\Cc}{{\mathcal{C}}}
\newcommand{\Nc}{{\mathcal{N}}}

\newcommand{\Dc}{{\mathcal{D}}}

\newcommand{\bx}{{\mathbf{x}}}

\newcommand{\bu}{{\mathbf{u}}}

\newcommand{\bc}{{\mathbf{c}}}
\newcommand{\bd}{{\mathbf{d}}}
\newcommand{\bk}{{\mathbf{k}}}

\newcommand{\bF}{{\mathbf{F}}}

\newcommand{\bD}{{\mathbf{D}}}

\newcommand{\bM}{{\mathbf{M}}}

\newcommand\xqed[1]{%
  \leavevmode\unskip\penalty9999 \hbox{}\nobreak\hfill
  \quad\hbox{#1}}
\newcommand\demo{\xqed{$\bullet$}}

\newcommand{\Var}{\operatorname{Var}}

\newcommand{\longthmtitle}[1]{\mbox{}\emph{(#1):}}
\newcommand{\setdef}[2]{\{#1 : #2\}}

\newcommand{\norm}[1]{\left\lVert#1\right\rVert}

\allowdisplaybreaks
\renewcommand{\baselinestretch}{.98}

\begin{document}
\title{\LARGE \bf Probabilistic Control Barrier Functions: Safety in Probability for Discrete-Time Stochastic Systems}
\author{Pol Mestres, Blake Werner, Ryan K. Cosner, Aaron D. Ames
\thanks{PM, BW, AA are with the Department of Mechanical and Civil
Engineering, California Institute of Technology, Pasadena, CA 91125, USA. RC is with the Department of Mechanical Engineering, Tufts University, Medford, MA 02155, USA.
Emails: \texttt{mestres,ames@caltech.edu}.
This research is supported by Technology Innovation Institute (TII).
}
}
\maketitle

\begin{abstract}
Control systems operating in the real world face countless sources of unpredictable uncertainties. These random disturbances can render deterministic guarantees inapplicable and cause catastrophic safety failures. To overcome this, this paper proposes a method for designing safe controllers for discrete-time stochastic systems that retain probabilistic guarantees of safety.
To do this we modify the traditional notion of a \textit{control barrier function} (CBF) to explicitly account for these stochastic uncertainties and call these new modified functions \textit{probabilistic CBFs}.
We show that probabilistic CBFs can be used to 
design controllers that guarantee safety over a finite number of time steps with a prescribed probability.
Next,
by leveraging various uncertainty quantification methods, such as concentration inequalities and the scenario approach, we provide a variety of sufficient conditions that result in computationally tractable controllers with tunable probabilistic guarantees across a plethora of practical scenarios. 
Finally, we showcase the applicability of our results in simulation and hardware for the control of a quadruped robot.
\end{abstract}

\section{Introduction}

Driven by applications in autonomous driving, robotics, and aerospace systems, the topic of safety has recently received a lot of attention in the control theory community.
Designing and verifying safe controllers in these domains is particularly challenging due to various sources of uncertainty, such as state estimation errors, imperfect perception, and unmodeled dynamics.
Although the study and design of safe controllers that are robust to such sources of uncertainty has received considerable attention in the literature, most existing works assume deterministic and bounded disturbances with known bounds.
Such assumptions are uncommon in practice and lead to overly conservative performance.
In contrast, this work is motivated by the need to design controllers that maintain safety in the presence of stochastic and potentially unbounded disturbances, conditions that naturally arise when employing techniques such as Kalman filtering or simultaneous localization and mapping (SLAM).

\subsection{Literature Review}

There are a variety of tools to design controllers with safety guarantees, including
control barrier functions (CBFs)~\cite{ADA-SC-ME-GN-KS-PT:19}, reachability-based controllers~\cite{SB-MC-SH-CJT:17}, and model predictive control~\cite{JBR-DQM-MMD:17}.
There has been a substantial effort to adapt these techniques for systems with uncertainty. In the CBF literature, this problem has been extensively studied for continuous-time systems and deterministic and bounded disturbances~\cite{SK-ADA:19,AA-AJT-CRH-GO-ADA:22,MJ:18,AA-TGM-ADA-GO:25,YW-XX:23}.
Other works have explored the problem under various uncertainty models, such as stochastic and possibly unbounded~\cite{AC:21-stochastic-cbfs,OS-CF:25,SP-AJ-GJP:04,CS-MD-SC:19,JS-RT:12},
Gaussian process-based~\cite{FC-JJC-BZ-CJT-KS:21,KL-VD-ML-JC-NA:22-ral} or through Wasserstein ambiguity sets~\cite{KL-YY-JC-NA:23-acc,PM-KL-NA-JC:24-csl}.
For discrete-time systems, the works~\cite{CS-MD-SC:19,RKC-PC-ADA:24,JS-RT:12} 
leverage the theory of martingales to provide bounds on the probability of exiting the safe set within a given time horizon. However, the condition on the control input derived in these papers only leverages the first moment of the distribution and can be conservative if more information about the uncertainty is available or not applicable if such first moment is unknown.
On the other hand,~\cite{MA-XX-ADA:22} introduces the risk-sensitive notion of CVaR safety, and a CBF-based condition that can be used to certify it.
However, controllers based on CVaR safety can be quite conservative in practice, and are only tractable for a small class of dynamics and safety constraints.
Similar risk-aware approaches are taken in~\cite{MB-GF-BH-DP-DP:22,MV-CP-JT:23}.
The recent work~\cite{AADN-AP-KM:24} leverages the scenario approach~\cite{MCC-SG-MP:09} to design controllers that are safe with a prescribed probability by forcing them to be safe for a number of different scenarios (i.e., uncertainty realizations).
Beyond CBF-based approaches, there exist various works in the literature that study the problem of guaranteeing safety constraints for discrete-time systems subject to stochastic uncertainty, particularly in the context of MPC~\cite{SP-TG-AZ-SS:22,KPW-LH-AC-MNZ:22,MF-LG-RS:16}.

\subsection{Statement of Contributions}

We consider the problem of designing safe controllers for discrete-time stochastic systems.
Our first contribution is the introduction of a probabilistic notion of CBF which we show can be used to design controllers with probabilistic safety guarantees over a finite time horizon.
However, verifying whether a function is a probabilistic CBF is impractical, because it requires knowledge of the uncertainty distribution.
In our second contribution we sidestep this difficulty by introducing two different sufficient conditions which guarantee that a function is a probabilistic CBF. 
These conditions only require knowledge of the moments of the distribution of the value of the CBF at the next time step.
The first condition only leverages a known upper bound on the CBF and the first moment of such distribution, whereas the second condition uses its first and second moments.
We also study various cases under which these conditions are convex in the control input and are thus amenable to optimization-based control design.
Next, we consider the problem of verifying that a function is a probabilistic CBF when moment information is not available but instead we have a dataset consisting of uncertainty realizations. 
In our third contribution we leverage concentration inequalities and scenario optimization to provide different sufficient conditions based on such dataset that ensure that a function is a probabilistic CBF with high confidence.
Finally, we validate our results in both simulation and hardware on quadruped robots.
Our main contributions are summarized in Table~\ref{tab:contribution-summary}.
\begin{table}[h!]
    \centering
    \begin{tabular}{|c|c|c|c|}
        \hline
        \textbf{Result} & 
        \textbf{Assumption} &
        \textbf{Technique} & \textbf{Guarantee} \\ \hline
        Prop.~\ref{prop:markov-based-condition} & 
        Mean
        & Markov's inequality & Single prob. \\ \hline
        Prop.~\ref{prop:probabilistic-CBFs-known-mean-variance} &
        Mean $+$ variance
        & Cantelli's inequality & Single prob.
        \\ \hline
        Prop.~\ref{prop:probabilistic-cbfs-give-confidence-hoeffding} 
        &
        Samples
        & Hoeffding's inequality & Double prob.
        \\ \hline
        Prop.~\ref{prop:probabilistic-CBFs-given-confidence-scenario-approach} & 
        Samples
        & Scenario approach & Double prob.
        \\ \hline
    \end{tabular}
    \caption{Table summarizing of the main results in the paper. The first column indicates the result,
    the second column the assumptions on the uncertainty needed to apply the result,
    the third column the main technical tool used in its derivation, and the fourth column indicates the type of guarantee it provides, which can be either single probability (as in Definition~\ref{def:delta-probabilistic-CBF}) or double probability (as in Definition~\ref{def:probabilistic-cbf-with-given-confidence}).
    }
    \label{tab:contribution-summary}
\end{table}

\subsection{Notation}

We denote by $\mathbb{N}, \real$ the set of natural and real numbers, respectively.
We use bold symbols to represent vectors and matrices and non-bold symbols to represent scalar quantities. 
For $N\in\mathbb{N}$, we let $[N]=\{1,\hdots,N\}$.
Given $\bx\in\real^n$, $\norm{\bx}$ denotes its Euclidean norm.
Given a matrix $\bM\in\real^{n\times n}$, $\text{Tr}(\bM)$ denotes its trace.
For a random variable $Z$, we denote by $\mathbb{E}[Z]$ and $\Var(Z)$ its expected value and variance, respectively.
If $Z$ is multivariate, 
$\text{Cov}(Z)$ denotes its covariance matrix.
Given $\delta > 0$, and $k\in\mathbb{N}$ values $R^{(1)}, \hdots, R^{(k)}$ sorted in non-decreasing order, $\text{Quantile}_{1-\delta}(R^{(1)}, \hdots, R^{(k)},\infty)$ denotes the $1-\delta$ quantile of the empirical distribution of the values $R^{(1)}, \hdots, R^{(k)}, \infty$, which can be equivalently obtained as $R^{(p)}$, where $p = \lceil (k+1)(1-\delta) \rceil$, and $\lceil \cdot \rceil$ denotes the ceiling function.


\section{Motivation}\label{sec:motivation}

Throughout this paper we 
consider a discrete-time stochastic system of the form:
\begin{align}\label{eq:discrete-time-system}
    \bx_{t+1} = \bF(\bx_t, \bu_t, \bd_t),
\end{align}
with $\bx_t\in\real^n$ the state, $\bu_t\in\real^m$ the control input, and $\bd_t\in\real^d$ a random variable distributed according to a distribution $\Dc$, i.e., $\bd_t \sim \Dc$.
For simplicity, we assume that $\Dc$ is state and time independent.

Given a user-prescribed safe set $\Cc\subset\real^n$, defined as the $0$-superlevel set of a continuous function $h$, (i.e., $\Cc = \setdef{\bx\in\real^n}{h(\bx) \geq 0}$),
our goal is to design a  
state-feedback controller $\bk:\real^n\to\real^m$ that 
ensures that the iterates of the 
closed-loop system obtained from~\eqref{eq:discrete-time-system} by using $\bk$ remain in $\Cc$ at all times.
However, the presence of the stochastic disturbance in~\eqref{eq:discrete-time-system} poses a significant challenge in finding such controller.
In fact, this is impossible if the distribution $\Dc$ is not supported on a bounded set, because the iterates will exit the set $\Cc$ in finite time with probability $1$ (cf.~\cite[Remark 1]{RKC-PC-AJT-ADA:23}).
Hence, we are interested in a notion of safety over a finite time horizon and with a prescribed probability, which we formalize next.
\begin{definition}\longthmtitle{Probabilistic safety over a finite time horizon}\label{def:probabilistic-safety-finite-time-horizon}
    Let $\epsilon\in(0,1)$, and $H\in\mathbb{N}$.
    We say that the controller $\bk:\real^n\to\real^m$ is $\epsilon$-safe over a time horizon $H$ if for any $\bx_0\in\Cc$, the iterates of~\eqref{eq:discrete-time-system} under $\bu_t = \bk(\bx_t)$ are such that 
    \begin{align*}
        \mathbb{P}\Big( \bigcap_{t=0}^H \{ \bx_t \in \Cc \} \Big) \geq 1-\epsilon.
    \end{align*}
\end{definition}

Next we introduce the notion of probabilistic CBF, which will prove to be a useful tool to certify probabilistic safety over a finite time horizon.
\begin{definition}\longthmtitle{Probabilistic CBF}\label{def:delta-probabilistic-CBF}
    Let $\delta\in(0,1)$.
    The function $h:\real^n\to\real$ is a $\delta$-probabilistic CBF if there exists $\alpha\in[0,1]$ such that
    for each $\bx\in\Cc$ there exists $\bu\in\real^m$ such that 
    \begin{align}\label{eq:probability-h-geq0}
        \mathbb{P}\Big( 
        h(\bF(\bx,\bu,\bd)) \geq \alpha h(\bx)
        \Big) \geq 1-\delta,
    \end{align}
\end{definition}
\noindent where $\bd\sim\Dc$.
We note that Definition~\ref{def:delta-probabilistic-CBF} is simply a probabilistic version of the standard CBF condition for discrete-time systems~\cite{AA-KS:17}.
In the sequel, whenever the parameter $\delta$ is clear from the context, we refer to a $\delta$-probabilistic CBF simply as a probabilistic CBF.
%
%
We also fix $\alpha\in[0,1]$, $\delta\in(0,1)$ and let $\Delta h:\real^n\times\real^m\times\real^d \to \real$ be defined as 
\begin{align*}
    \Delta h(\bx,\bu,\bd) = h(\bF(\bx,\bu,\bd)) - \alpha h(\bx).
\end{align*}
The following result shows that probabilistic CBFs can be used to verify the notion of probabilistic safety over a given time horizon in Definition~\ref{def:probabilistic-safety-finite-time-horizon}.

\begin{proposition}\longthmtitle{Probabilistic guarantees over a finite time horizon}\label{prop:probabilistic-guarantees-finite-time-horizon}
    Let $\bx_0\in\Cc$, $H\in\mathbb{N}$, and $\bk:\real^n\to\real^m$ be such that $\bu = \bk(\bx)$ satisfies~\eqref{eq:probability-h-geq0} for each $\bx\in\Cc$.
    Then, for the system~\eqref{eq:discrete-time-system} under $\bu_t = \bk(\bx_t)$, we have 
    \begin{align}\label{eq:prob-intersection-1-delta-H}
        \mathbb{P}\Big( \bigcap_{t=0}^H \{ \bx_t \in \Cc \} \Big) \geq (1-\delta)^H.
    \end{align}
    In particular, for any $\epsilon > 0$, if 
    \begin{align}\label{eq:delta-condition}
        \delta \leq 1 - (1-\epsilon)^{\frac{1}{H}},
    \end{align}
    then $\bk$ is $\epsilon$-safe over a time horizon $H$.
\end{proposition}
\begin{proof}
    Define the event $\Gamma_{t} := \bigcap_{s=0}^t \{ \bx_s \in \Cc \}$.
    By the law of conditional probabilities~\cite[page 63]{SMR:10}
    \begin{align}\label{eq:product-conditional-probs}
        \mathbb{P}\Big(
        \bigcap_{t=0}^H \{ \bx_t \in \Cc \}
        \Big) = \prod_{t=1}^H \mathbb{P}\Big( 
        \bx_t \in \Cc | \Gamma_{t-1} \Big).
    \end{align}
    Note that for each $t\in[H]$, 
    \begin{align*}
        &\mathbb{P}\Big( 
        \bx_t \in \Cc | \Gamma_{t-1} \Big) = \\
        &\mathbb{P}\Big( 
        h(\bF(\bx_{t-1},\bk(\bx_{t-1}),\bd_{t-1})) \geq 0 | \Gamma_{t-1}
        \Big) \geq \\
        &\mathbb{P}\Big( 
        h(\bF(\bx_{t-1},\bk(\bx_{t-1}),\bd_{t-1})) \geq \alpha h(\bx_{t-1}) | \Gamma_{t-1}
        \Big) \geq 1-\delta,
    \end{align*}
    The first inequality follows from the fact that if $\bx_{t-1}\in\Cc$, then $h(\bx_{t-1}) \geq 0$, and 
    $h(\bF(\bx_{t-1},\bk(\bx_{t-1}),\bd_{t-1}) \geq \alpha h(\bx_{t-1}))$ implies that $h(\bF(\bx_{t-1},\bk(\bx_{t-1}),\bd_{t-1}) \geq 0$. 
    The last inequality follows from~\eqref{eq:probability-h-geq0}.
    Hence, by~\eqref{eq:product-conditional-probs} we have that~\eqref{eq:prob-intersection-1-delta-H} holds.
    By taking $\delta$ to satisfy~\eqref{eq:delta-condition}, it follows that $(1-\delta)^{H} \geq 1-\epsilon$, and $\bk$ is $\epsilon$-safe over a time horizon $H$.
\end{proof}

Proposition~\ref{prop:probabilistic-guarantees-finite-time-horizon} shows that controllers obtained from the probabilistic CBF condition~\eqref{eq:probability-h-geq0}
can be used to verify the notion of probabilistic safety introduced in Definition~\ref{def:probabilistic-safety-finite-time-horizon}.
Unfortunately, verifying that $h$ is a probabilistic CBF as per Definition~\ref{def:delta-probabilistic-CBF} is impractical because it requires computing a probability over the distribution $\Dc$, which is often unknown.
In the following sections, we devise various sufficient conditions that can be used to verify Definition~\ref{def:delta-probabilistic-CBF} without requiring full knowledge of the distribution $\Dc$.


\section{Verification of Probabilistic CBFs using known Moments}\label{sec:known-random-variable-moments}

In practice, although one generally does not know the full distribution $\Dc$, it is common to have information about some of its moments.
In this section we study how to leverage this information to verify that $h$ is a probabilistic CBF.
Our first result establishes a sufficient condition for that to hold by imposing a constraint on the expected value of  $\Delta h$.
\begin{proposition}\longthmtitle{Markov-based condition}\label{prop:markov-based-condition}
    Suppose that there exists $b > 0$ such that
    for all $\bx\in\Cc$ and $\bu\in\real^m$,
    $\Delta h(\bx,\bu,\bd) \leq b$ almost surely.
    Furthermore, suppose that for all $\bx\in\Cc$, there exists $\bu\in\real^m$ such that
    \begin{align}\label{eq:expectation-condition-markov}
        \mathbb{E}[ \Delta h( \bx,\bu,\bd ) ] \geq b(1-\delta).
    \end{align}
    Then, $h$ is a $\delta$-probabilistic CBF.
\end{proposition}
\begin{proof}
    Note that by assumption, for each $\bx\in\Cc$ there exists $\bu\in\real^m$ and $b  > 0$ such that the random variable $b - \Delta h(\bx,\bu,\bd)$ is positive almost surely. 
    By Markov's inequality~\cite[Theorem 5.11]{AK:13},
    \begin{align*}
        \mathbb{P}( \Delta h(\bx,\bu,\bd) \leq 0 ) &= \mathbb{P}( b  - \Delta h( \bx,\bu,\bd ) \geq b  ) \\
        &\leq \frac{ \mathbb{E}[ b-\Delta h( \bx,\bu,\bd ) ] }{ b }.
    \end{align*}
    Now, from~\eqref{eq:expectation-condition-markov} we have that $\mathbb{E}[b-\nabla h(\bx,\bu,\bd)] \leq b\delta$ and hence
    we have for each $\bx\in\Cc$ there exists $\bu\in\real^m$ such that $\mathbb{P}( \Delta h(\bx,\bu,\bd) \leq 0 ) \leq \delta$, from where it follows that $h$ is a $\delta$-probabilistic CBF.
\end{proof}


Since~\eqref{eq:expectation-condition-markov} is derived though Markov's inequality, we refer to it as the \textbf{Markov-based condition}.
Proposition~\ref{prop:markov-based-condition} shows that even if the distribution $\Dc$ is unknown, the expected values of the random variable $\Delta h(\bx,\bu,\bd))$ can be used to verify that $h$ is a probabilistic CBF.
Note also that one simple condition that ensures that $\Delta h$ is uniformly upper bounded is if $h$ is uniformly upper bounded.

In practice, in order to design a controller that ensures probabilistic safety, we are interested in including input constraints such as~\eqref{eq:expectation-condition-markov} as constraints of an optimization problem.
However, 
as it is also the case for deterministic discrete-time CBFs~\cite{AA-KS:17},
the condition~\eqref{eq:expectation-condition-markov} is generally not convex in $\bu$, which compromises the computational tractability of such optimization-based control design.
The following result 
leverages the convexity properties of $h$ to derive another set of conditions for which the convexity properties with respect to $\bu$ are easier to analyze.
\begin{corollary}\longthmtitle{Markov-based condition under convexity assumptions}\label{cor:use-of-convexity}
    Suppose that there exists $b > 0$ such that for all $\bx\in\Cc$ and $\bu\in\real^m$, 
    $\Delta h(\bx,\bu,\bd) \leq b$ almost surely.
    Further assume that one of the following conditions holds:
    \begin{itemize}
        \item $h$ is convex and for all $\bx\in\Cc$ there exists $\bu\in\real^m$ such that 
        \begin{align}\label{eq:markov-expectation-jensen}
            \tilde{\Delta}h(\bx,\bu) := h( \mathbb{E}[ \bF(\bx,\bu,\bd) ] ) -\alpha h(\bx) \geq b(1-\delta);
        \end{align}
        \item $h$ is concave, $\lambda > 0$ satisfies $\sup\limits_{\bx\in\real^n} \norm{\nabla^2 h(\bx)} \leq \lambda$ and for all $\bx\in\Cc$, there exists $\bu\in\real^m$ such that
        \begin{align}\label{eq:markov-expectation-jensen-gap}
            \hspace{-0.3cm}
            &\hat{\Delta} h(\bx,\bu) \! := \! \tilde{\Delta}h(\bx,\bu) \! - \! 
            \frac{\lambda}{2}\text{Tr}( \text{Cov}(\bF(\bx,\bu,\bd) ) ) \! \geq \! b(1\!-\!\delta),
        \end{align}
    \end{itemize}
    Then, $h$ is a $\delta$-probabilistic CBF.
\end{corollary}
\begin{proof}
    First, if $h$ is convex, by Jensen's inequality, we have $h( \mathbb{E}[ \bF(\bx,\bu,\bd) ] ) \leq \mathbb{E}[ h( \bF(\bx,\bu,\bd) ) ]$. Therefore, the satisfaction of~\eqref{eq:markov-expectation-jensen} implies that~\eqref{eq:expectation-condition-markov} holds. By Proposition~\ref{prop:markov-based-condition}, this means that $h$ is a $\delta$-probabilistic CBF.
    Second, if $h$ is concave, by~\cite[Lemma 1]{RKC-PC-ADA:24} (which is based on a result from~\cite{RAB:12}), we have that
    \begin{align}
        \notag
        &\mathbb{E}[ h(\bF(\bx,\bu,\bd)) ] \geq \\
        &h( \mathbb{E}[\bF(\bx,\bu,\bd) ]) - 
        \frac{\lambda}{2}\text{Tr}( \text{Cov}(\bF(\bx,\bu,\bd) ) ).
    \end{align}
    Hence, the satisfaction of~\eqref{eq:markov-expectation-jensen-gap} implies that~\eqref{eq:expectation-condition-markov} holds. By Proposition~\ref{prop:markov-based-condition}, this means that $h$ is a $\delta$-probabilistic CBF.
\end{proof}

The following remark comments on the convexity properties with respect to $\bu$ of the conditions~\eqref{eq:markov-expectation-jensen} and~\eqref{eq:markov-expectation-jensen-gap}.
\begin{remark}\longthmtitle{Convexity properties of constraints~\eqref{eq:markov-expectation-jensen},~\eqref{eq:markov-expectation-jensen-gap}}\label{rem:convexity-optimization-based-control-constraints}
    For general $\bF$ and $h$, conditions~\eqref{eq:markov-expectation-jensen} and~\eqref{eq:markov-expectation-jensen-gap} are not convex in $\bu$.
    In fact, even if $\bF$ is affine in $\bu$, and $h$ is convex,~\eqref{eq:markov-expectation-jensen} is not convex in $\bu$.
    However, in some cases, such as when $h$ is quadratic,~\eqref{eq:markov-expectation-jensen} is a quadratic concave constraint in $\bu$.
    Although it can not be included as a constraint in a convex program, there exist very efficient heuristics to solve non-convex quadratically constrained quadratic programs (QCQPs)~\cite{JP-SB:17-arxiv}.
    %
    %
    On the other hand,~\eqref{eq:markov-expectation-jensen-gap} is convex in $\bu$ if $\bF$ is affine in $\bu$ and $\bd$ is additive, (i.e. $\bF(\bx,\bu,\bd) = F_1(\bx) + F_2 \bu + \bd$, for $F_1:\real^n\to\real^n$ and $F_2\in\real^{n\times m}$) and $h$ is concave.
    Additionally, if $\text{Trace}( \text{Cov}(h(\bF(\bx,\bu,\bd) ])) ) \leq c$ for some $c>0$, $h$ is concave and $\bF$ is affine in $\bu$, then the condition that we get after replacing $\frac{\lambda}{2}\text{Tr}( \text{Cov}(\bF(\bx,\bu,\bd) ) )$ by $\frac{\lambda c}{2}$ in~\eqref{eq:markov-expectation-jensen-gap} is convex in $\bu$.
    %
    \demo
\end{remark}

The following result provides another sufficient condition for $h$ to be a probabilistic CBF.
In this case, the condition does not require a known upper bound for $\Delta h$ but assumes knowledge of its variance (i.e., its second moment).
\begin{proposition}\longthmtitle{Cantelli-based condition}\label{prop:probabilistic-CBFs-known-mean-variance}
    Suppose that for each $\bx\in\Cc$ there exists $\bu\in\real^m$ such that
    \begin{align}\label{eq:mean-var-inequality}
        \mathbb{E}[ \Delta h(\bx,\bu,\bd) ] - \sqrt{ \Var\Big( \Delta h(\bx,\bu,\bd) \Big) \frac{1-\delta}{\delta} } \geq 0.
    \end{align}
    Then, $h$ is a $\delta$-probabilistic CBF.
\end{proposition}
\begin{proof}
    By Cantelli's inequality~\cite{FPC:28}, 
    we have that for any $\bx\in\real^n, \bu\in\real^m$, and $\kappa \geq 0$, 
    \begin{align*}
        &\mathbb{P}\Big( 
        \Delta h(\bx,\bu,\bd) \leq \mathbb{E}[\Delta h(\bx,\bu,\bd)] - \kappa
        \Big) \\
        &\leq \frac{ \Var( \Delta h(\bx,\bu,\bd) )  }{ \Var( \Delta h(\bx,\bu,\bd)) + \kappa^2 }
    \end{align*}
    By taking $\kappa = \kappa_{\bx,\bu} := \sqrt{ \Var\Big( \Delta h(\bx,\bu,\bd) \Big) \frac{1-\delta}{\delta} }$, we have 
    \begin{align*}
        \mathbb{P}\Big( 
        \Delta h(\bx,\bu,\bd) \leq \mathbb{E}[\Delta h(\bx,\bu,\bd)] - \kappa_{\bx,\bu}
        \Big) \leq \delta.
    \end{align*}
    Now, for each $\bx\in\Cc$, select $\bu\in\real^m$ so that~\eqref{eq:mean-var-inequality} holds.
    It follows that 
    \begin{align*}
        &\mathbb{P}\Big(
        \Delta h(\bx,\bu,\bd) \leq 0
        \Big) \leq \\
        &\mathbb{P}\Big(
        \Delta h(\bx,\bu,\bd) \leq \mathbb{E}[ \Delta h(\bx,\bu,\bd) ] - \kappa_{\bx,\bu}
        \Big) \leq \delta.
    \end{align*}
    Hence, $\mathbb{P}\Big( \Delta h(\bx,\bu,\bd) \geq 0 \Big) \geq 1-\delta$,
    and $h$ is a $\delta$-probabilistic CBF.
\end{proof}

Since~\eqref{eq:mean-var-inequality} is derived through Cantelli's inequality, we refer to it as the \textbf{Cantelli-based condition.}
Proposition~\ref{prop:probabilistic-CBFs-known-mean-variance} shows that even if the distribution $\Dc$ is unknown, if the expected value and variance of $\Delta h$ are known, then one can verify that $h$ is a probabilistic CBF.
The following result is an analogue of Corollary~\ref{cor:use-of-convexity} for the Cantelli-based condition.
\begin{corollary}\longthmtitle{Use of convexity properties for Cantelli-based condition}\label{cor:use-of-convexity-for-condition-based-on-known-mean-and-variance}
    Let $\tilde{\Delta} h$ and $\hat{\Delta} h$ be defined as in Corollary~\ref{cor:use-of-convexity} and
    assume that one of the following conditions holds:
    \begin{itemize}
        \item $h$ is convex and for all $\bx\in\Cc$ there exists $\bu\in\real^m$ such that 
        \begin{align}\label{eq:mean-var-ineq-jensen}
        \hspace{-0.5cm}
            \tilde{\Delta}h(\bx,\bu) \geq \sqrt{ \Var\Big( \Delta h(\bx,\bu,\bd) \Big) \frac{1-\delta}{\delta} };
        \end{align}
        \item $h$ is concave, $\lambda > 0$ satisfies $\sup\limits_{\bx\in\real^n} \norm{\nabla^2 h(\bx)} \leq \lambda$ and for all $\bx\in\Cc$, there exists $\bu\in\real^m$ such that
        \begin{align}\label{eq:mean-var-ineq-jensen-gap}
            &\hat{\Delta}h(\bx,\bu) \geq \sqrt{ \Var\Big( \Delta h(\bx,\bu,\bd) \Big) \frac{1-\delta}{\delta} }.
        \end{align}
    \end{itemize}
    Then, $h$ is a $\delta$-probabilistic CBF.
\end{corollary}
The proof of Corollary~\ref{cor:use-of-convexity-for-condition-based-on-known-mean-and-variance} follows an argument analogous to that of Corollary~\ref{cor:use-of-convexity}.

\begin{remark}\longthmtitle{Convexity of the associated constraints for known mean and variance}\label{rem:convexity-associated-constraints-known-mean-variance}
    Comments similar to those made in Remark~\ref{rem:convexity-optimization-based-control-constraints} follow regarding conditions~\eqref{eq:mean-var-inequality},~\eqref{eq:mean-var-ineq-jensen},~\eqref{eq:mean-var-ineq-jensen-gap}.
    However, in this case the presence of the term $\Var( \Delta h(\bx,\bu,\bd) )$ can further complicate the convexity guarantees for~\eqref{eq:mean-var-ineq-jensen},~\eqref{eq:mean-var-ineq-jensen-gap}.
    Nonetheless, in the case where 
    $\bd$ is additive and $h$ is affine, 
    $\Var( \Delta h(\bx,\bu,\bd) ) = \Var( h(\bd) )$, and~\eqref{eq:mean-var-ineq-jensen},~\eqref{eq:mean-var-ineq-jensen-gap} have the same dependence in $\bu$ as the one discussed in Remark~\ref{rem:convexity-optimization-based-control-constraints}.
    Additionally, in the case where $\Delta h$ is affine in $\bu$~\eqref{eq:mean-var-inequality} is a second order cone (SOC) constraint, which is convex and can be easily used in a second-order cone program (SOCP) as in~\cite{FC-JJC-BZ-CJT-KS:21,PM-JC:24-auto}.
    \demo
\end{remark}



The results in Propositions~\ref{prop:markov-based-condition} and~\ref{prop:probabilistic-CBFs-known-mean-variance}, leverage knowledge of the first and second moments of the distribution of $\Delta h$ to provide bounds on the probability of its tails.
Although it is possible to derive similar (possibly tighter) results when knowledge of higher moments is available (cf.~\cite{KI:64,LV-SB-KC:07,BBB:87} for concentration inequalities using higher moments), this usually leads to conditions with complicated dependencies on the control input $\bu$ and that are computationally not amenable for optimization-based control.

\section{Verification of Probabilistic CBFs Using Data}\label{sec:confidence-based-probabilistic-cbfs}

Section~\ref{sec:known-random-variable-moments}
established a set of results that show how to verify that a function is a probabilistic CBF by leveraging moments of the uncertainty distribution.
In this section, we take an alternative approach that relies on the availability of samples of the uncertainty distribution, rather than knowledge of its moments.
However, since the samples themselves are random, 
in general it is not possible to exactly certify that $h$ is a probabilistic CBF.
Instead, one can only do that with a given confidence. 
The following definition formalizes this idea.
\begin{definition}\longthmtitle{Probabilistic CBF with a given confidence}\label{def:probabilistic-cbf-with-given-confidence}
    Let $\delta,\beta\in(0,1)$, and
    $\bD_N = \{ \bd^{(i)} \}_{i=1}^N$ be $N\in\mathbb{N}$ independent identically distributed samples with $\bd^{(i)} \sim \Dc$.
    The function $h:\real^n\to\real$ is a $\delta$-probabilistic CBF with confidence $1-\beta$ if for each $\bx\in\Cc$, there exists $\bu\in\real^m$ (dependent on the dataset $\bD_N$) such that 
    \begin{align}\label{eq:probabilistic-CBF-given-confidence}
        \mathbb{P}_N \Big( 
        \mathbb{P}( \Delta h(\bx,\bu, \bd) \geq 0 ) \geq 1-\delta
        \Big) \geq 1-\beta,
    \end{align}
    where $\mathbb{P}_N$ denotes the probability with respect to $\bD_N$.
\end{definition}

The following result shows that the notion of CBF in
Definition~\ref{def:probabilistic-cbf-with-given-confidence} can be used to verify that safety holds over a finite time horizon with a given probability over the stochasticity of~\eqref{eq:discrete-time-system} and confidence over $\bD_N$.

\begin{proposition}\longthmtitle{Safety over a finite time horizon for probabilistic CBFs with given confidence}\label{prop:safety-over-finite-time-horizon-prob-cbfs-given-confidence}
    Let $\epsilon,\tilde{\beta}\in(0,1)$.
    Let $\bD_N = \{ \bd^{(i)} \}_{i=1}^N$ be $N$ independent identically distributed samples with $\bd^{(i)}\sim\Dc$.
    Let $\delta,\beta\in(0,1)$, be such that~\eqref{eq:delta-condition} and $\beta \leq \frac{\tilde{\beta}}{H}$ hold.
    Further assume that $h$ is a $\delta$-probabilistic CBF with confidence  $1-\beta$. 
    Let $\bk:\real^n\to\real^m$ be such that $\bu_t = \bk(\bx_t)$ satisfies~\eqref{eq:probabilistic-CBF-given-confidence} for each $\bx_t\in\Cc$.
    Then, $\bk$ is $\epsilon$-safe over a horizon $H$ with confidence $1-\tilde{\beta}$, i.e., for any $\bx_0\in\Cc$,
    the iterates of~\eqref{eq:discrete-time-system} under $\bk$ satisfy
    \begin{align}\label{eq:finite-horizon-safety-probability-prob-CBF-given-confidence}
        \mathbb{P}_N\Big( 
        \mathbb{P}\Big( 
        \bigcap_{t=0}^H \{ \bx_t\in\Cc \}
        \Big) \geq 1-\epsilon
        \Big) \geq 1-\tilde{\beta}.
    \end{align}
\end{proposition}
\begin{proof}
    As shown in the proof of Proposition~\ref{prop:probabilistic-guarantees-finite-time-horizon} and using the same notation, for any dataset $\bD_N$, we have
    \begin{align*}
        \mathbb{P}\Big( 
        \bigcap_{t=0}^H \{ & \bx_t\in\Cc \}
        \Big) \! \\
        & \geq \! \prod_{t=1}^H
        \mathbb{P}\Big(
        \Delta h(\bx_{t-1},\bk(\bx_{t-1}),\bd_{t-1}) \geq 0 | \Gamma_{t-1}
        \Big).
    \end{align*}
    Hence, if we let $A_N$ be the event
    \begin{align*}
        \prod_{t=1}^H
        \mathbb{P}\Big(
        \Delta h(\bx_{t-1},\bk(\bx_{t-1}),\bd_{t-1}) \geq 0 | \Gamma_{t-1}
        \Big)
        \geq 1-\epsilon,
    \end{align*}
    (defined by the randomness of the dataset $\bD_N$) we have 
    \begin{align}\label{eq:first-inequality}
        \mathbb{P}_N\Big( 
        \mathbb{P}\Big( 
        \bigcap_{t=0}^H \{ \bx_t\in\Cc \}
        \Big) \geq 1-\epsilon
        \Big) \geq \mathbb{P}_N (A_N).
    \end{align}
    Now note that if the dataset $\bD_N$ is such that the event $B_{t,N}$, defined as
    \begin{align*}
        \mathbb{P}\Big(
        \Delta h(\bx_{t-1},\bk(\bx_{t-1}),\bd_{t-1}) \geq 0 | \Gamma_{t-1}
        \Big) \geq (1-\epsilon)^{\frac{1}{H}},
    \end{align*}
    holds for all $t\in[H]$, then $A_N$ also holds.
    Therefore,
    \begin{align}\label{eq:second-inequality}
        &\mathbb{P}_N( A_N )
        \geq \mathbb{P}_N\Big(
        \bigcap_{t=1}^H 
        B_{t,N}
        \Big) \geq \sum_{t=1}^H \mathbb{P}_N (B_{t,N}) - (H-1),
    \end{align}
    where the last inequality follows from Fréchet's inequality~\cite{MF:35}.
    Now, since $1-\delta \geq (1-\epsilon)^{\frac{1}{H}}$,
    if the event $C_{t,N}$, defined as
    \begin{align*}
        \mathbb{P}\Big(
        \Delta h(\bx_{t-1},\bk(\bx_{t-1}),\bd_{t-1}) \geq 0 | \Gamma_{t-1}
        \Big) \geq 1-\delta,
    \end{align*}
    holds, then $B_{t,N}$ also holds.
    Therefore, $\mathbb{P}_N(B_{t,N}) \geq \mathbb{P}_N(C_{t,N})$.
    Now, since $h$ is a $\delta$-probabilistic CBF with confidence $\beta$, we have $\mathbb{P}_N(C_{t,N}) \geq 1-\beta$ for all $t\in[H]$,
    and~\eqref{eq:first-inequality},~\eqref{eq:second-inequality} we have
    \begin{align*}
        \mathbb{P}_N\Big( 
        \mathbb{P}\Big( 
        \bigcap_{t=0}^H \{ \bx_t\in\Cc \}
        \Big) \geq 1-\epsilon
        \Big) \geq  
        1-H\beta.
    \end{align*}
    Now the result follows from the fact that $\beta \leq \frac{\tilde{\beta}}{H}$.
\end{proof}

In the rest of the section, we fix $\delta,\beta\in(0,1)$,
and consider a dataset $\bD_N$ of $N\in\mathbb{N}$ samples as defined in Proposition~\ref{prop:safety-over-finite-time-horizon-prob-cbfs-given-confidence}.
The following result establishes a sufficient condition for $h$ to be a probabilistic CBF with a given confidence in the case where the random variable $\Delta h(\bx,\bu,\bd)$ is bounded and has a special separable form.

\begin{proposition}\longthmtitle{Hoeffding-based condition}\label{prop:probabilistic-cbfs-give-confidence-hoeffding}
    Suppose that there exist functions $H_1,H_3:\real^n\times\real^m\to\real$, $H_2,H_4:\real^n\times\real^d\to\real$ such that 
    \begin{align}\label{eq:separability-condition}
        \Delta h(\bx,\bu,\bd) = H_1(\bx,\bu) + H_2(\bx,\bd) + H_3(\bx,\bu) H_4(\bx,\bd).
    \end{align}
    Suppose that for each $\bx\in\Cc$, there exists $\bu\in\real^m$ such that:
    \begin{itemize}
        \item there exist constants $a, b, a_2, b_2, a_4, b_4, b_3$ satisfying $a \leq \Delta h(\bx,\bu,\bd) \leq b$, $a_2 \leq H_2(\bx,\bd) \leq b_2$, $a_4 \leq H_4(\bx,\bd) \leq b_4$, and $|H_3(\bx,\bu)| \leq b_3$ almost surely;
        \item there exist $\epsilon, \epsilon_1, \epsilon_2 > 0$ such that $\epsilon_1 + b_3 \epsilon_2 \leq \epsilon$, $\frac{\beta}{2} \geq 2 \exp \{ -\frac{2 N \epsilon_1^2 }{  (b_2-a_2)^2 } \}$, 
        $\frac{\beta}{2} \geq 2 \exp \{ -\frac{2 N \epsilon_2^2 }{  (b_4-a_4)^2 } \}$, and 
        \begin{align}
            &\frac{1}{N} \sum_{i=1}^N \Delta h(\bx,\bu,\bd^{(i)}) \geq \epsilon + b(1-\delta),~\label{eq:hoeffding-mean} 
        \end{align}
    \end{itemize}
    Then, $h$ is a $\delta$-probabilistic CBF with confidence $1-\beta$.
\end{proposition}
\begin{proof}
    Let $\bx\in\Cc$ and pick $\bu\in\real^m$ so that
    the conditions in the statement hold.
    Note that such $\bu$ is dependent on the dataset $\bD_N$.
    Now, we have
    \begin{align*}
        &\mathbb{P}_N\Big( 
        \Big\rvert\mathbb{E}[\Delta h(\bx,\bu,\bd)] - \frac{1}{N} \sum_{i=1}^N \Delta h(\bx,\bu,\bd^{(i)}) 
        \Big\rvert
        \leq \epsilon
        \Big) \\
        &=\mathbb{P}_N\Big(
        \Big\rvert
        \tilde{H}_2(\bx) + H_3(\bx,\bu)\tilde{H}_4(\bx)
        \Big\rvert \leq \epsilon
        \Big),
    \end{align*}
    where 
    \begin{align*}
        \tilde{H}_2(\bx) = \mathbb{E}[\Delta H_2(\bx,\bd)] - \frac{1}{N} \sum_{i=1}^N H_2(\bx,\bd^{(i)}), \\
        \tilde{H}_4(\bx) = \mathbb{E}[H_4(\bx,\bd)] - \frac{1}{N} \sum_{i=1}^N H_4(\bx,\bd^{(i)}).
    \end{align*}
    Now, let $A_{H,2}$ be the event that $|\tilde{H}_2(\bx)| \leq \epsilon_1$. Since $\frac{\beta}{2} \geq 2 \exp \{ -\frac{2 N \epsilon_1^2 }{  (b_2-a_2)^2 } \}$, by Hoeffding's inequality~\cite{WH:63-hoeffding} we have that 
    $\mathbb{P}_N(A_{H,2}) \geq 1-\frac{\beta}{2}$.
    Similarly, let $A_{H,4}$ be the event that 
    $|\tilde{H}_4(\bx)| \leq \epsilon_2$.
    Since $\frac{\beta}{2} \geq 2 \exp \{ -\frac{2 N \epsilon_2^2 }{  (b_4-a_4)^2 } \}$, by Hoeffding's inequality~\cite{WH:63-hoeffding} we have that 
    $\mathbb{P}_N(A_{H,4}) \geq 1-\frac{\beta}{2}$.
    Therefore, 
    $\mathbb{P}_N(A_{H,2}\cap A_{H,4}) \geq 1-\beta$.
    Now, note that if the event $A_{H,2}\cap A_{H,4}$ holds, since $\epsilon_1 + b_3 \epsilon_2 \leq \epsilon$, we have that 
    \begin{align*}
        &\Big\rvert\mathbb{E}[\Delta h(\bx,\bu,\bd)] - \frac{1}{N} \sum_{i=1}^N \Delta h(\bx,\bu,\bd^{(i)}) 
        \Big\rvert
        \leq \epsilon
    \end{align*}
    also holds. Therefore, 
    \begin{align*}    
        &\mathbb{P}_N\Big( 
        \rvert\mathbb{E}[\Delta h(\bx,\bu,\bd)] - \frac{1}{N} \sum_{i=1}^N \Delta h(\bx,\bu,\bd^{(i)}) 
        \Big\rvert \leq \epsilon
        \Big) \geq 1-\beta.
    \end{align*}
    This, together with~\eqref{eq:hoeffding-mean}, implies that
    \begin{align*}
        \mathbb{P}_N\Big( 
        \mathbb{E}[ \Delta h(\bx,\bu,\bd) ] \geq b(1-\delta)
        \Big) \geq 1-\beta.
    \end{align*}
    Now, the result follows from Proposition~\ref{prop:markov-based-condition}.
\end{proof}

Since~\eqref{eq:hoeffding-mean} is derived through Hoeffding's inequality, we refer to it as the \textbf{Hoeffding-based condition.}

\begin{remark}\longthmtitle{Conditions of Proposition~\ref{prop:probabilistic-cbfs-give-confidence-hoeffding}}
    The separable form of $\Delta h$ in~\eqref{eq:separability-condition} enables the use of Hoeffding's inequality in the proof. Indeed, if~\eqref{eq:separability-condition} does not hold, since $\bu$ is possibly dependent on $\bD_N$, the quantities $\{ \Delta h(\bx,\bu,\bd^{(i)}) \}_{i=1}^N$ are not necessarily independent and Hoeffding's inequality can not be applied.
    On the other hand, a simple scenario in which $a, b$ as in Proposition~\ref{prop:probabilistic-cbfs-give-confidence-hoeffding} exist is if $h$ is uniformly bounded (which we can always assume without loss of generality). 
    Moreover, $\frac{\beta}{2} \geq 2 \exp \{ -\frac{2 N \epsilon_1^2 }{  (b_2-a_2)^2 } \}$, 
    $\frac{\beta}{2} \geq 2 \exp \{ -\frac{2 N \epsilon_2^2 }{  (b_4-a_4)^2 } \}$ hold if $N$ is sufficiently large.
    We also note that a version of Proposition~\ref{prop:probabilistic-cbfs-give-confidence-hoeffding} can be stated in which the constants 
    defined therein are not uniform and instead depend on $\bx$ or $\bu$.
    Although the assumption that such constants are uniform across $\bx$ and $\bu$ might add some conservatism, it leads to a simpler dependency on $\bu$ in~\eqref{eq:hoeffding-mean}.
    In fact,~\eqref{eq:hoeffding-mean} is convex in $\bu$ under assumptions such as the ones in Remark~\ref{rem:convexity-optimization-based-control-constraints}.
    \demo
\end{remark}

Next we introduce another result that certifies that $h$ is a probabilistic CBF with a given confidence.
In this case, the result is distribution-free (meaning it holds without any assumptions on the uncertainty distribution) and uses the theory of the scenario approach~\cite{MCC-SG-MP:09}.
\begin{proposition}\longthmtitle{Scenario approach-based condition}\label{prop:probabilistic-CBFs-given-confidence-scenario-approach}
    Suppose that for each $\bx\in\Cc$, there exists $\bu\in\real^m$ such that 
    \begin{align}\label{eq:scenario-based-condition}
        \Delta h(\bx,\bu,\bd^{(i)}) \geq 0, \ i\in[N].
    \end{align}
    Further suppose that for each $\bx\in\Cc$ and $\bd\in\real^d$, the function $\bu \to -\Delta h(\bx,\bu,\bd)$ is convex and 
    $N$ satisfies 
    \begin{align}\label{eq:scenario-condition-N-choose-i}
        \sum_{i=0}^{d-1} \begin{pmatrix}
            N \\ i
        \end{pmatrix} \delta^i (1-\delta)^{N - i } \leq \beta.
    \end{align}
    Then, $h$ is a $\delta$-probabilistic CBF with confidence $1-\beta$.
\end{proposition}
\begin{proof}
    Let $\bc\in\real^m$ and take $\bu_{\bx}\in\real^m$ as a minimizer of the optimization problem
    \begin{align}\label{ref:scenario-lp}
        &\min\limits_{\bu\in\real^m} \bc^\top \bu, \quad \text{s.t.} \ \Delta h(\bx,\bu,\bd^{(i)}) \geq 0, \ i\in[N],
    \end{align}
    which is feasible by assumption.
    Now, by~\cite[Theorem 2.4]{MCC-SG:08}, for each $\bx\in\Cc$, 
    $\mathbb{P}_N\Big( \mathbb{P}(\Delta h(\bx,\bu,\bd) \leq 0 ) > \delta \Big) \leq \beta$.
    Equivalently, 
    \begin{align}\label{eq:scenario-inequality-aux}
        \mathbb{P}_N\Big( 
        \mathbb{P}(\Delta h(\bx,\bu,\bd) \leq 0 ) \leq \delta 
        \Big) \geq 1-\beta.
    \end{align}
    Now, note that for any $\bD_N$, if the event $X_N$ defined as $\mathbb{P}(\Delta h(\bx,\bu,\bd) \leq 0 ) \leq \delta$ holds, then the event $Y_N$, defined as $\mathbb{P}(\Delta h(\bx,\bu,\bd) \geq 0) \geq 1-\delta$ also holds.
    Hence, $\mathbb{P}_N(Y_N) \geq \mathbb{P}_N(X_N)$, and the result follows from~\eqref{eq:scenario-inequality-aux}.
\end{proof}

Since condition~\eqref{eq:scenario-based-condition} is derived through the scenario approach, we refer to it as the \textbf{scenario-based condition}.
The following remark discusses the assumptions in Proposition~\ref{prop:probabilistic-CBFs-given-confidence-scenario-approach}.
\begin{remark}\longthmtitle{Assumptions of Proposition~\ref{prop:probabilistic-CBFs-given-confidence-scenario-approach}}\label{rem:prob-cbfs-scenario}
    As mentioned in~\cite[Theorem 1]{MCC-SG-MP:09}, a sufficient condition for~\eqref{eq:scenario-condition-N-choose-i} to hold is $N \geq \frac{2}{\delta}\Big( \ln\frac{1}{\beta} + d \Big)$.
    Also, although Proposition~\ref{prop:probabilistic-CBFs-given-confidence-scenario-approach} requires $-\Delta_{\alpha}$ to be convex in $\bu$, the recent paper~\cite{SG-MCC:24} shows that a similar result can be obtained in the non-convex setting.
    \demo
\end{remark}

\section{Experimental Validation}

In this section we validate our theoretical results by considering the problem of controlling a quadrupedal robot traversing a narrow corridor, inspired by~\cite[Section 5.D]{RKC-PC-AJT-ADA:23}.
We consider the following discrete-time dynamics
\begin{align}\label{eq:quadruped-dynamics}
    \bx_{k+1} = \bx_k + \Delta t \begin{pmatrix}
        \cos(\theta) & -\sin(\theta) & 0 \\
        \sin(\theta) & \cos(\theta) & 0 \\
        0 & 0 & 1
    \end{pmatrix}
    \begin{pmatrix}
        v_k^x \\ v_k^y \\ \omega_k
    \end{pmatrix} + \bd_k,
\end{align}
where $\bx_k = [x, y, \theta]^\top$ and $\Delta t > 0$.
We let $h(\bx) = 0.5^2 - y^2$.

\subsection{Simulation}
In simulation, we model the uncertainties in the terrain through a Gaussian disturbance with standard deviation $\sigma$ acting only in the $y$ direction, i.e., $\bd_k = [0, d_k, 0]$, with $d_k \sim \Nc(0,\sigma)$.
We consider a time horizon of 20 steps and take $\delta = 0.1$, $\beta = 0.01$, 
$\alpha = 0.01$
and $\Delta t = 0.1$.
We consider a nominal controller $\bk_{\text{nom}}(\bx) = [0.2, 0, -\theta]$. For each of the different conditions ——Markov (in the form of~\eqref{eq:markov-expectation-jensen-gap}), Cantelli, Hoeffding (note that the condition~\eqref{eq:separability-condition} holds in this example), and scenario approach——, we compute the control input at every state through a safety filter, i.e., by solving an optimization problem that finds the closest input to the nominal controller that satisfies the corresponding condition.
We generate datasets of size 3252 and 113, 
for the Hoeffding and scenario approaches, respectively.
For the Markov, Hoeffding and scenario approach-based condition, the optimization problem at every state is a convex quadratically constrained quadratic program (QCQP), which we solve using the \texttt{cvxpy} library in Python.
The Cantelli-based condition leads to a non-convex problem, which we solve using the trust region solver of the \texttt{minimize} function in the \texttt{SciPy} library in Python.
We implement each of the approaches (along with the nominal controller and the approach based on~\cite{RKC-PC-AJT-ADA:23}, which we refer to as the Martingale approach) 
in 400 different trajectories with initial condition $\bx_0 = [0, 0, 0]$.

Figure~\ref{fig:plot_probabilities} shows the 
number of unsafe trajectories for different values of $\sigma$ (and fixed horizon $20$) and different time horizon values (and fixed $\sigma$).
For a fixed time horizon of 20
Proposition~\ref{prop:probabilistic-guarantees-finite-time-horizon} 
ensures that the probability that the trajectory is safe over 20 steps is at least $(1-\delta)^{20} \approx 0.12$, and hence the probability that the trajectory is unsafe over these 20 steps is at most $0.88$.
(for the approaches in Section~\ref{sec:confidence-based-probabilistic-cbfs}, this is with confidence $1-\beta \cdot 20 = 0.8$).
As can be seen in Figure~\ref{fig:plot_probabilities}, all the approaches considered in this paper lead to a fraction of unsafe trajectories significantly lower than the theoretical bound, which suggests that the conditions might be conservative.
The source of this conservatism can be attributed to the limited knowledge of the uncertainty distribution, the use of Jensen's inequality (cf. Corollaries~\ref{cor:use-of-convexity} and~\ref{cor:use-of-convexity-for-condition-based-on-known-mean-and-variance}), and the fact that these results do not consider full trajectory information, and instead derive probabilistic bounds only based on a condition at every state.

\begin{figure}
    \centering
    \includegraphics[width=0.98\linewidth]{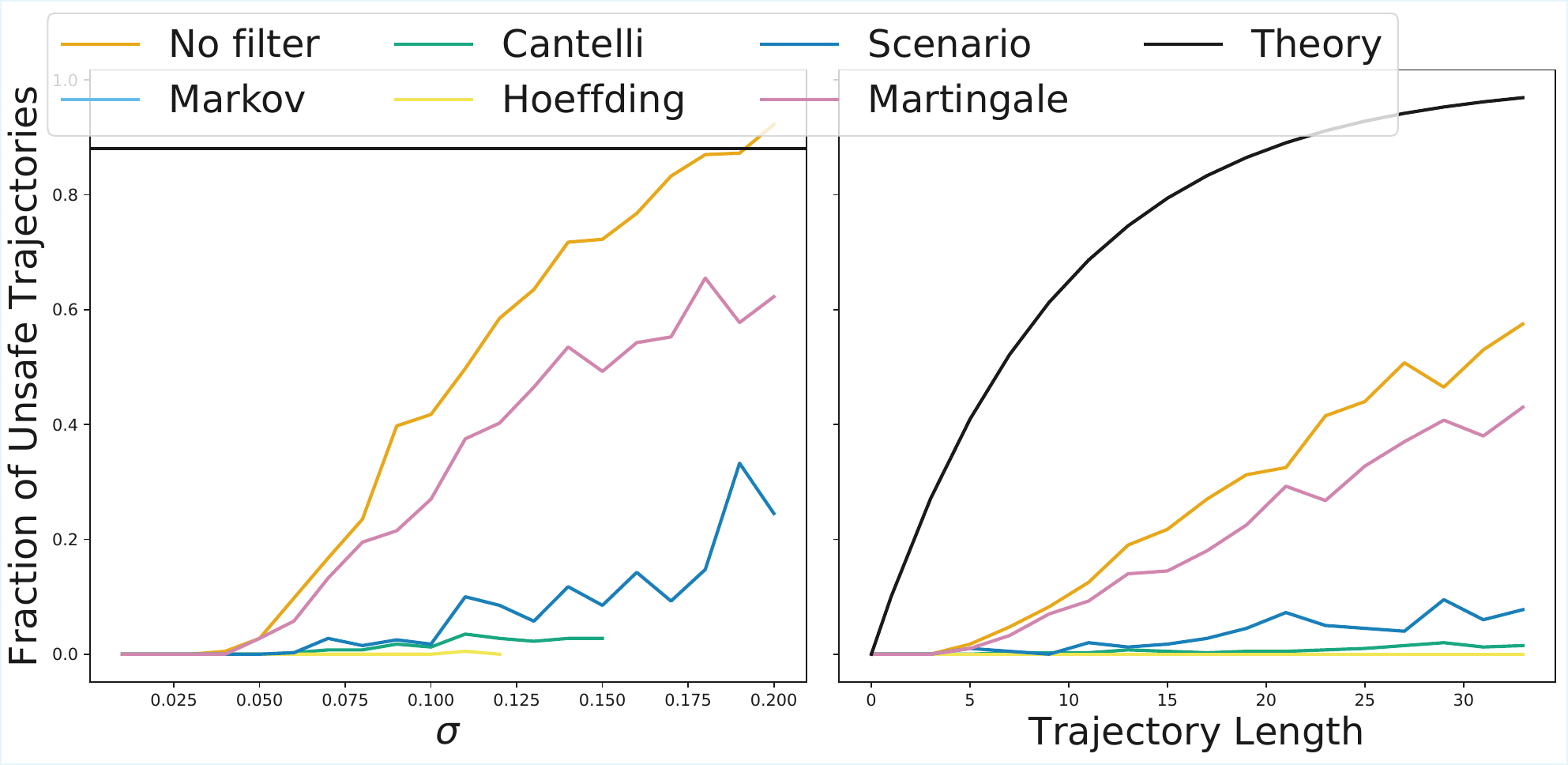}
    \caption{(left) Number of unsafe trajectories for different values of the disturbance variance $\sigma$. (right) Number of unsafe trajectories for different values of the trajectory length.
    No filter refers to the implementation of the nominal controller, whereas Markov, Cantelli, Hoeffding, and Scenario
    refer to the approaches in Propositions~\ref{prop:markov-based-condition},~\ref{prop:probabilistic-cbfs-give-confidence-hoeffding},~\ref{prop:probabilistic-CBFs-given-confidence-scenario-approach},
    respectively.}
    \label{fig:plot_probabilities}
\end{figure}

Figure~\ref{fig:trajs-all-approaches} shows 400 different trajectories generated with the different approaches for $\sigma = 0.06$ and trajectory length $20$. Note that we do not plot trajectories from the Markov approach because it is infeasible for such value of $\sigma$. 

\begin{figure*}[t]
    \centering
    \begin{subfigure}{0.19\textwidth}
        \centering
        \includegraphics[width=0.95\linewidth]{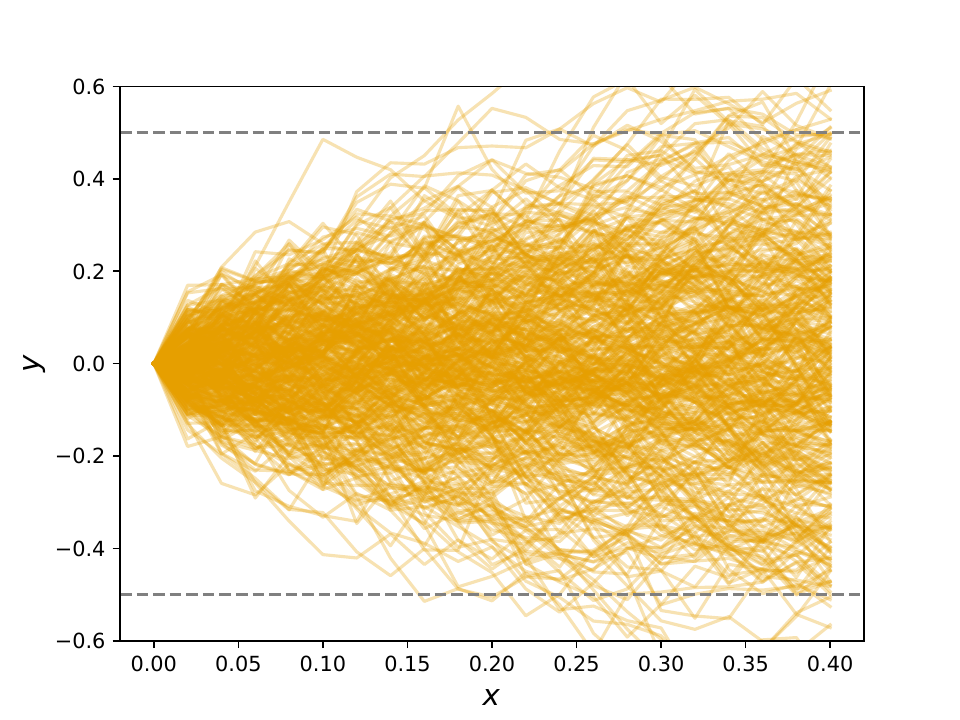}
        \caption{No filter}
        \label{fig:no-filter}
    \end{subfigure}
    \begin{subfigure}{0.19\textwidth}
        \centering
        \includegraphics[width=0.95\linewidth]{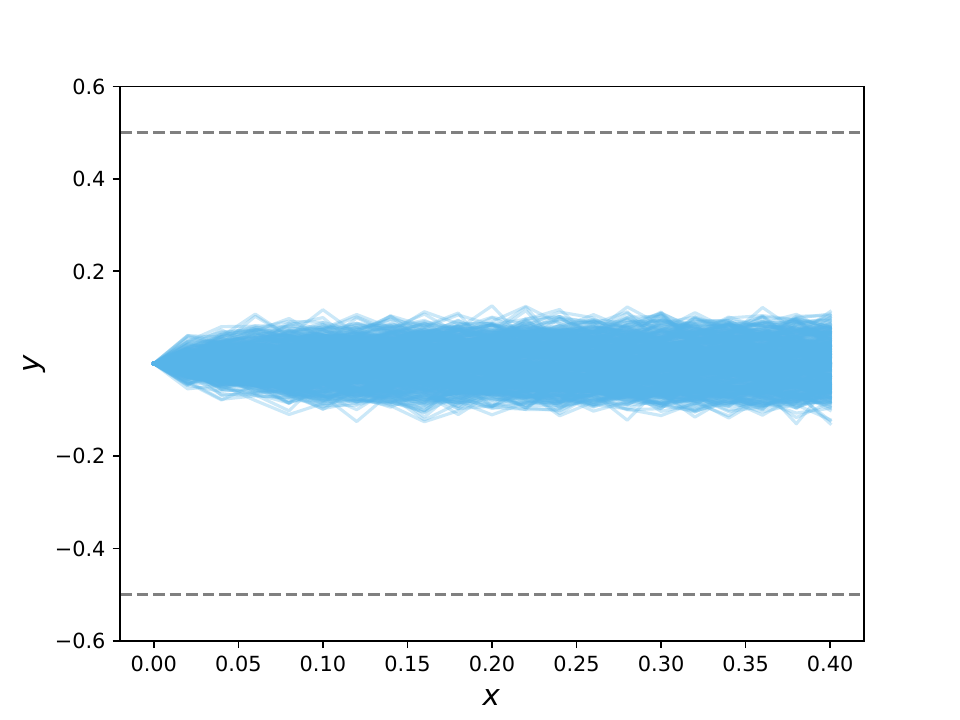}
        \caption{Markov}
        \label{fig:no-filter}
    \end{subfigure}
    \begin{subfigure}{0.19\textwidth}
        \centering
        \includegraphics[width=0.95\linewidth]{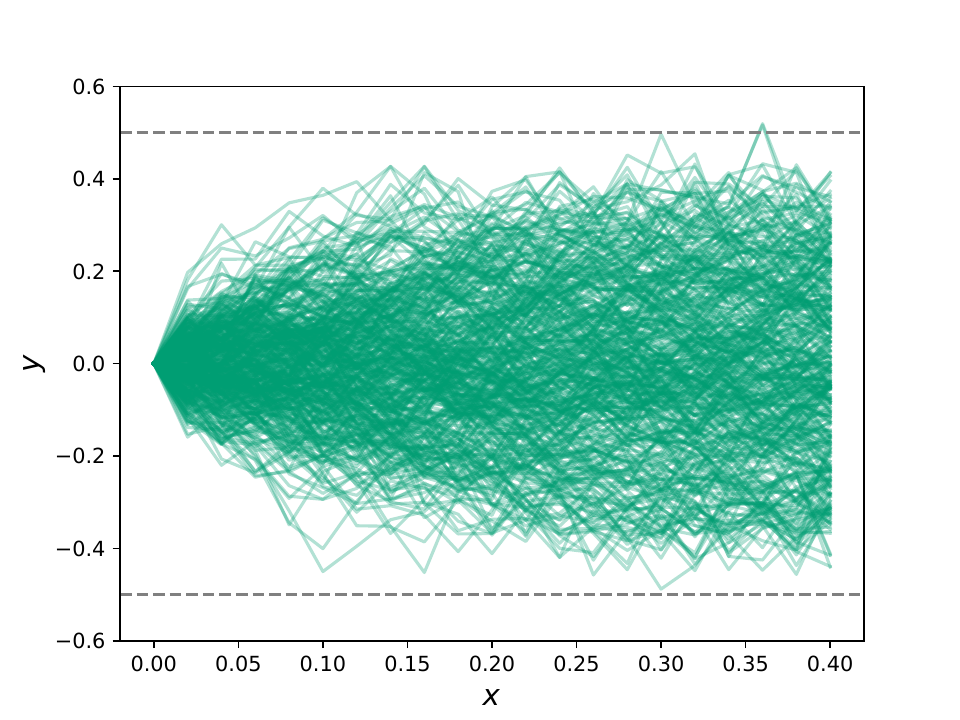}
        \caption{Cantelli}
        \label{fig:cantelli}
    \end{subfigure}
    \begin{subfigure}{0.19\textwidth}
        \centering
        \includegraphics[width=0.95\linewidth]{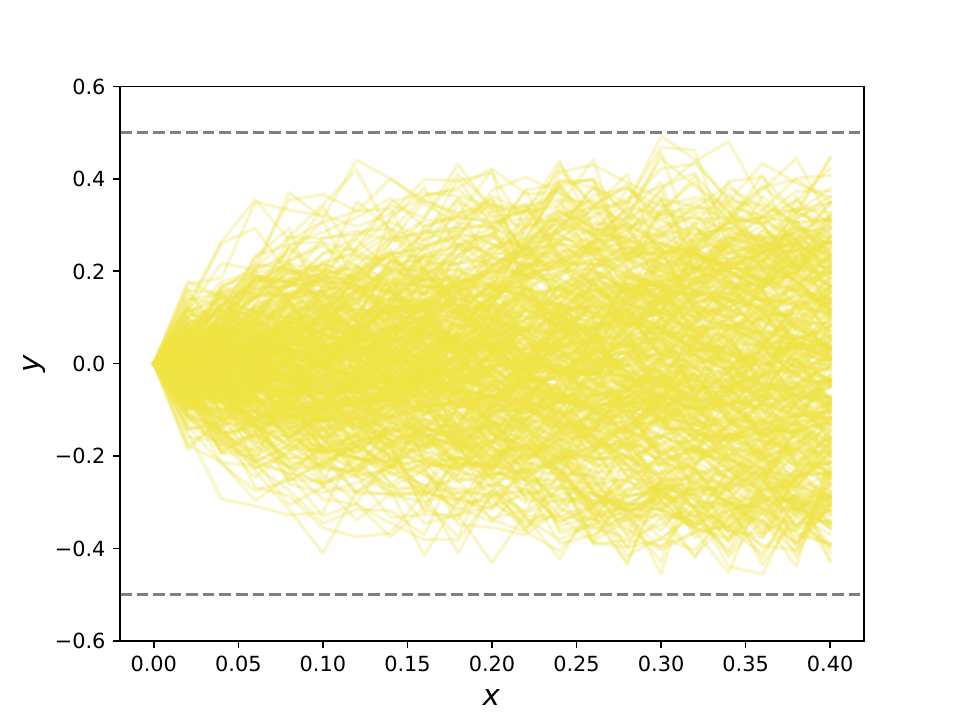}
        \caption{Hoeffding}
        \label{fig:hoeffding}
    \end{subfigure}
    \begin{subfigure}{0.19\textwidth}
        \centering
        \includegraphics[width=0.95\linewidth]{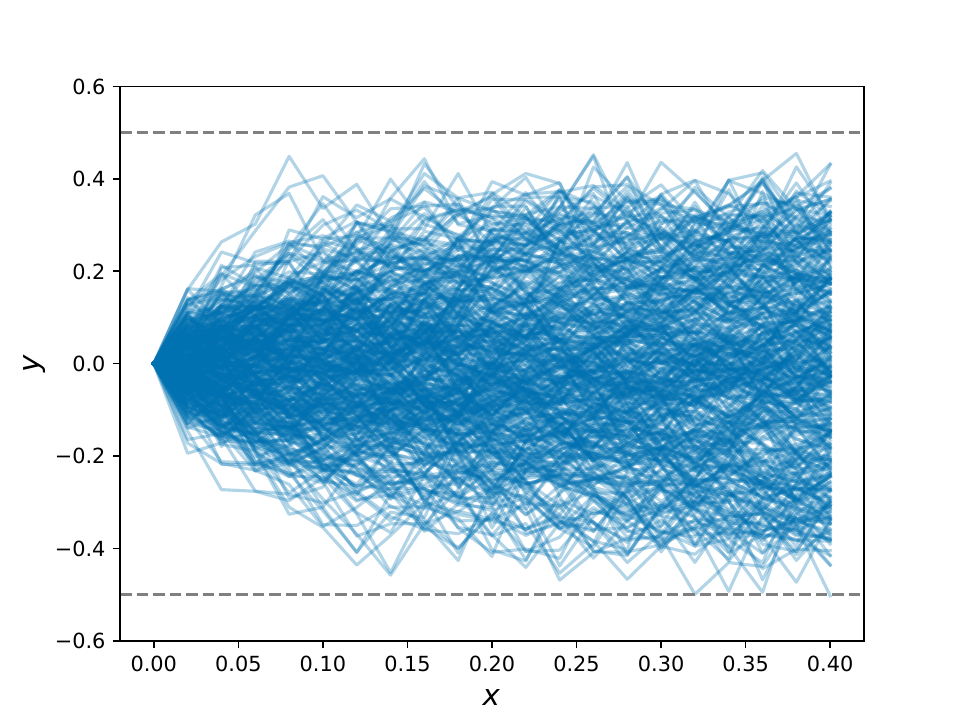}
        \caption{Scenario}
        \label{fig:scenario}
    \end{subfigure}

    \caption{Evolution of 400 different trajectories with initial condition at $\bx_0 = [0, 0, 0]$ and $\sigma = 0.06$ ($\sigma = 0.02$ for Markov) and trajectory length $20$ for different methods introduced in the paper, along with the nominal controller (left).}
    \label{fig:trajs-all-approaches}
\end{figure*}

\begin{table}[h!]
    \centering
    \begin{tabular}{|@{}c@{}|c|c|c|c|c|}
        \hline
        & Markov & Cantelli & Hoeffding & Scenario & Martingale \\ \hline
        $T_{ex}$ & 3.47 & 3.78 & 3.89 & 79.96 & 3.46 \\ \hline
        $\sigma_0$ & 0.03 & 0.16 & 0.13 & 0.21 & 0.25 \\ \hline
    \end{tabular}
    \caption{Average time it takes to compute the controller ($T_{ex}$) in milliseconds and smallest value of $\sigma$ that leads to infeasibility $\sigma_0$ for the different approaches.
    }
    \label{tab:exec-times-infeas}
\end{table}

Finally, Table~\ref{tab:exec-times-infeas} shows the average execution times for the different approaches as well as the smallest value of $\sigma$
that leads to infeasibility of the corresponding condition.

The choice of which of the different approaches should be used depends on a variety of factors. Firstly, on 
what information is available about the uncertainty distribution.
Secondly, on the desired level of conservativeness with respect to the safety constraints. Figure~\ref{fig:plot_probabilities} and Table~\ref{tab:exec-times-infeas} show that there is a tradeoff between the empirical number of safety violations and
the smallest variance that leads to infeasibility.
Approaches yielding less safety violations generally lead to infeasible conditions for smaller values of $\sigma$.
Thirdly, another point to consider is the computational time (cf. Table~\ref{tab:exec-times-infeas}). 
Although the martingale approach shows a larger value of $\sigma_0$ in Table~\ref{tab:exec-times-infeas}, it fails to meet the prescribed safety guarantees that the other approaches successfully satisfy.








\subsection{Hardware}

\begin{figure}
    \centering
    \includegraphics[width=0.85\linewidth]{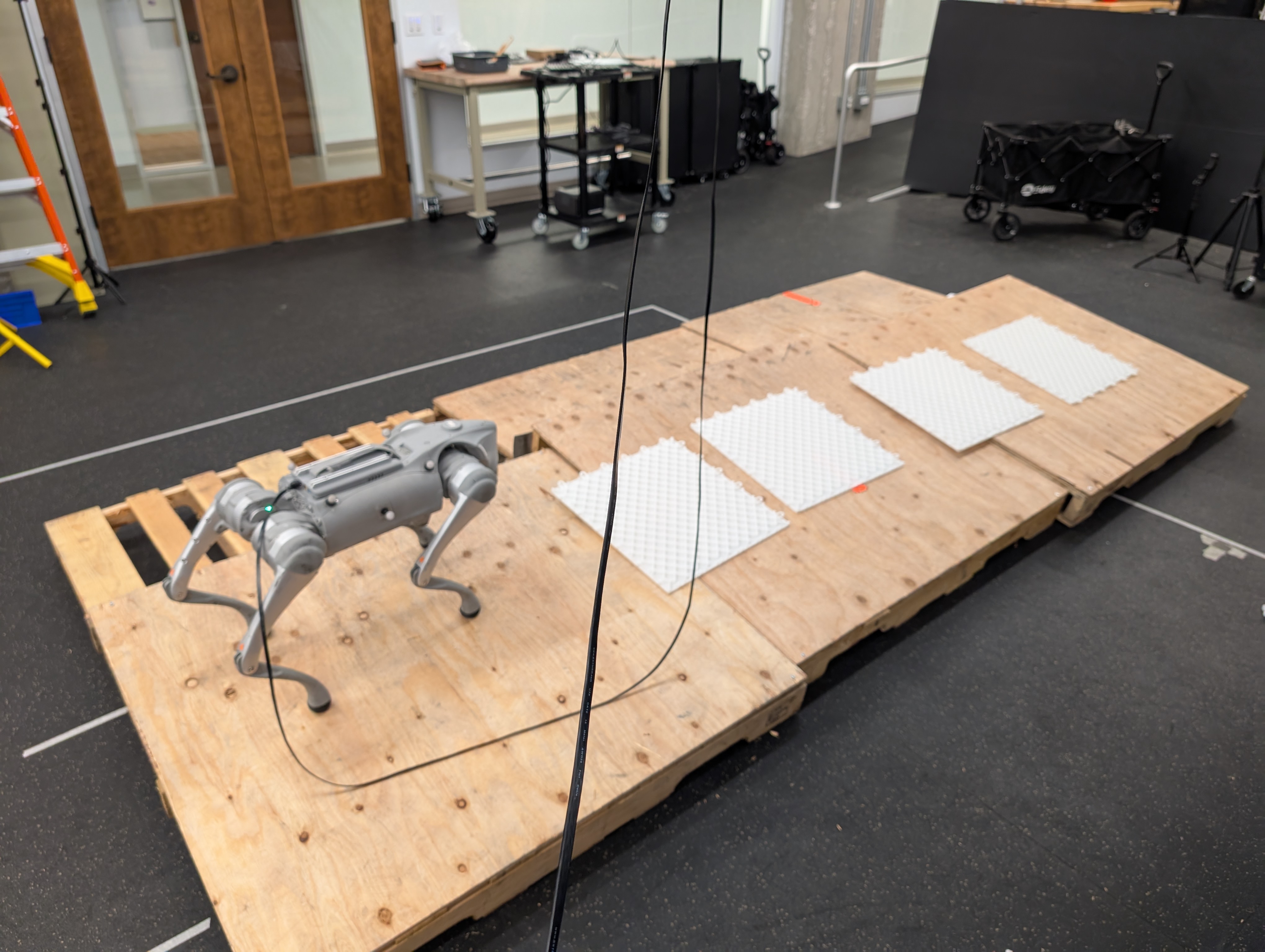}
    \caption{Experimental Setup}
    \label{fig:experimental_setup}
\end{figure}
To validate the effectiveness of our approach beyond simulation, we deploy it on a quadruped robot navigating a challenging obstacle course designed to induce substantial disturbances. The course is 3m long and 1m wide, with a laterally slanted wooden ramp covered in slippery plates (Fig. \ref{fig:experimental_setup}). Our system is a Unitree GO2 quadruped, which we assume follows dynamics~\eqref{eq:quadruped-dynamics}.
The control objective is to traverse the course with a nominal forward velocity of 1 m/s.
We compare three controllers: (i) the nominal controller with no safety filter, (ii) a naive CBF-based filter, and (iii) our Hoeffding-based probabilistic CBF filter~\footnote{Video of the experiment: https://youtu.be/Lu9tl-teSLc}. For the latter, we collect disturbance data from 5000 rollout steps on the course to characterize the uncertainty. 
We use the same $h$ as in simulation, and estimate the global state of the robot by an external OptiTrack system.
\begin{figure}
    \centering
    \includegraphics[width=0.9\linewidth]{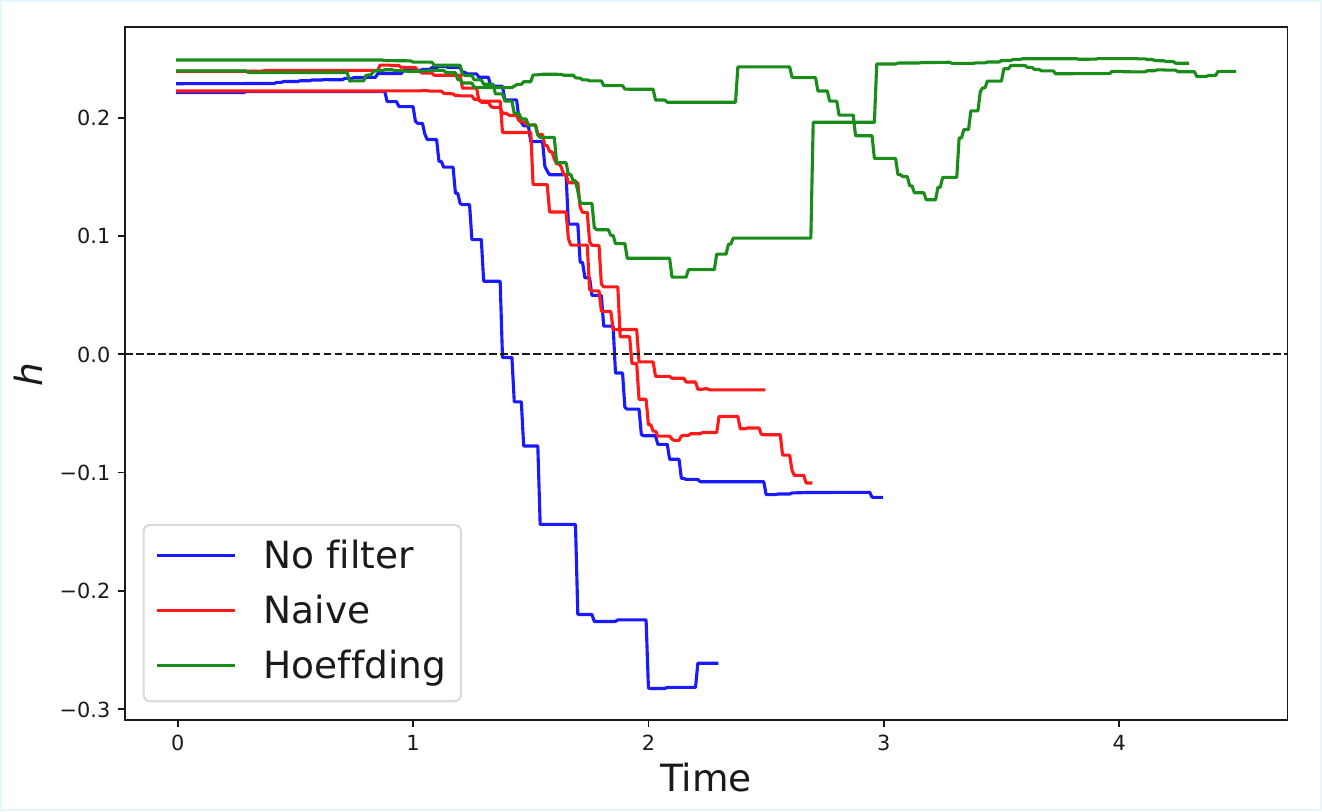}
    \caption{Evolution of the value of $h$ for two different trajectories for each of the different approaches considered in the hardware experiment.}
    \label{fig:hardware_trajectories}
\end{figure}
Results are shown in Fig. \ref{fig:hardware_trajectories}. 
We plot the trajectories induced by the nominal controller, the 
naive filter (i.e., the controller that is computed at each state as the smallest deviation from the nominal satisfying the discrete-time CBF condition assuming the disturbance is identically zero), and the Hoeffding-based approach.
The nominal controller and the naive filter
consistently violate the safety constraint as the sloped terrain pushes the robot laterally out of the safe set. In contrast, our Hoeffding-based filter substantially reduces the frequency and magnitude of constraint violations, correctly accounting for the stochastic disturbance and succeeding on 4/5 of the tested rollouts.

\section{Conclusion}
We have introduced probabilistic CBFs, which can be used to design controllers that satisfy safety constraints with a prescribed probability over a finite time horizon.
We have proposed a number of different sufficient conditions to verify that a function is a probabilistic CBF. The first set of conditions leverage knowledge of the moments of the uncertainty distribution, whereas the second set of conditions utilize realizations of the uncertainty.
We have studied the computational tractability of such conditions for its use in optimization-based control. In particular, we have implemented such controllers on a quadruped, both in simulation and hardware.
In future work, we plan to find less conservative conditions to ensure that a function is a probabilistic CBF and bridge the gap between the theoretical probabilistic safety guarantees and the empirical results.
The results in this paper also open the door to the design of safe controllers
in systems with state estimation errors induced by techniques such as the Kalman Filter or SLAM.

\section*{Acknowledgments}
The authors thank Ryan M. Bena for his help with the quadruped hardware experiment.

\bibliography{bib/alias,bib/Main-add,bib/Main,bib/JC,bib/New}
\bibliographystyle{IEEEtran}

\end{document}